\documentclass{article}
\usepackage{arxiv}
\usepackage[utf8]{inputenc} 
\usepackage[T1]{fontenc}    
\usepackage{doi}
\usepackage{amsthm}

\theoremstyle{definition}
\newtheorem{definition}{Definition}

\theoremstyle{plain}
\newtheorem{theorem}{Theorem}

\theoremstyle{plain}
\newtheorem{corollary}{Corollary}

\theoremstyle{definition}
\newtheorem{proposition}{Proposition}

\theoremstyle{definition}
\newtheorem{lemma}{Lemma}

\theoremstyle{definition}
\newtheorem{conjecture}{Conjecture}

\theoremstyle{remark}
\newtheorem{remark}{Remark}

\usepackage{graphicx}
\usepackage[all]{xy}
\usepackage{amsmath,amssymb}

\usepackage{caption}
\usepackage{mathtools}

\usepackage[inline]{enumitem}
\usepackage{ifthen}
\usepackage[utf8]{inputenc} 
\usepackage{xcolor}

\usepackage[capitalize]{cleveref}

\usepackage{quiver}

\protected\def\dg{\ensuremath{\times}}
\DeclareUnicodeCharacter{05A8}{\dg}

\usepackage{tikz}

\usepackage{amssymb}
\usepackage{mathtools}
\usepackage{newpxtext}
\usepackage[varg,bigdelims]{newpxmath}
\usepackage{mathrsfs}
\usepackage{dutchcal}
\usepackage{fontawesome}
\usepackage{listings}
\usepackage[export]{adjustbox}

  \crefformat{enumi}{\##2#1#3}
  \crefalias{chapter}{section}

  \setlist{nosep}
  \setlistdepth{6}

\definecolor{dgreen}{rgb}{0.0, 0.5, 0.3} 
\definecolor{dyellow}{rgb}{8.0, 0.74, 0}

  \usetikzlibrary{ 
  	cd,
  	math,
  	decorations.markings,
		decorations.pathreplacing,
  	positioning,
  	arrows.meta,
  	shapes,
		shadows,
		shadings,
  	calc,
  	fit,
  	quotes,
  	intersections,
    circuits,
    circuits.ee.IEC
  }
  
  \tikzset{trees/.style={
	inner sep=0, 
	minimum width=0, 
	minimum height=0,
	level distance=.5cm, 
	sibling distance=.5cm,
	edge from parent/.style={shorten <= \outdist, draw, ->},
	grow'=up,
	decoration={markings, mark=at position 0.75 with \arrow{stealth}}
	},
     out distance/.store in=\outdist,
      out distance={0pt},
 }

    \tikzset{
     oriented WD/.style={
        every to/.style={out=0,in=180,draw},
        label/.style={
           font=\everymath\expandafter{\the\everymath\scriptstyle},
           inner sep=0pt,
           node distance=2pt and -2pt},
        semithick,
        node distance=1 and 1,
        decoration={markings, mark=at position \stringdecpos with \stringdec},
        ar/.style={postaction={decorate}},
        execute at begin picture={\tikzset{
           x=\bbx, y=\bby,
           }}
        },
     string decoration/.store in=\stringdec,
     string decoration={\arrow{stealth};},
     string decoration pos/.store in=\stringdecpos,
     string decoration pos=.7,
	 	 dot size/.store in=\dotsize,
	   dot size=3pt,
	 	 dot/.style={
			 circle, draw, thick, inner sep=0, fill=black, minimum width=\dotsize
	   },
     bbx/.store in=\bbx,
     bbx = 1.5cm,
     bby/.store in=\bby,
     bby = 1.5ex,
     bb port sep/.store in=\bbportsep,
     bb port sep=1.5,
     bb port length/.store in=\bbportlen,
     bb port length=4pt,
     bb penetrate/.store in=\bbpenetrate,
     bb penetrate=0,
     bb min width/.store in=\bbminwidth,
     bb min width=1cm,
     bb rounded corners/.store in=\bbcorners,
     bb rounded corners=2pt,
     bb small/.style={bb port sep=1, bb port length=2.5pt, bbx=.4cm, bb min width=.4cm, 
bby=.7ex},
		 bb medium/.style={bb port sep=1, bb port length=2.5pt, bbx=.4cm, bb min width=.4cm, 
bby=.9ex},
     bb/.code 2 args={
        \pgfmathsetlengthmacro{\bbheight}{\bbportsep * (max(#1,#2)+1) * \bby}
        \pgfkeysalso{draw,minimum height=\bbheight,minimum width=\bbminwidth,outer 
sep=0pt,
           rounded corners=\bbcorners,thick,
           prefix after command={\pgfextra{\let\fixname\tikzlastnode}},
           append after command={\pgfextra{\draw
              \ifnum #1=0{} \else foreach \i in {1,...,#1} {
                 ($(\fixname.north west)!{\i/(#1+1)}!(\fixname.south west)$) +(-
\bbportlen,0) 
  coordinate (\fixname_in\i) -- +(\bbpenetrate,0) coordinate (\fixname_in\i')}\fi 
              \ifnum #2=0{} \else foreach \i in {1,...,#2} {
                 ($(\fixname.north east)!{\i/(#2+1)}!(\fixname.south east)$) +(-
\bbpenetrate,0) 
  coordinate (\fixname_out\i') -- +(\bbportlen,0) coordinate (\fixname_out\i)}\fi;
           }}}
     },
     bb name/.style={append after command={\pgfextra{\node[anchor=north] at 
(\fixname.north) {#1};}}}
  }

\tikzset{
  tick/.style={postaction={
    decorate,
    decoration={markings, mark=at position 0.5 with {\draw[-] (0,.4ex) -- (0,-.4ex);}}}
  },
  tickx/.style={
    postaction={ decorate,
      decoration={markings,
        mark=at position 0.5 with {
          \fill circle [radius=.28ex];
        }
      }
    }
  }
}

\tikzset{
  ttick/.style={postaction={
    decorate,
    decoration={markings, mark=at position 0.47 with {\draw[-] (0,.4ex) -- (0,-.4ex);}, mark=at position 0.53 with {\draw[-] (0,.4ex) -- (0,-.4ex);}
    }
  }
  },
  tickx/.style={
    postaction={ decorate,
      decoration={markings,
        mark=at position 0.5 with {
          \fill circle [radius=.28ex];
        }
      }
    }
  }
}

\tikzset{
  swish/.style={postaction={
    decorate,
    decoration={markings, mark=at position 0.2 with {\rotatebox{90}{$\sim$}}}}
  },
  swishx/.style={
    postaction={ decorate,
      decoration={markings,
        mark=at position 0.2 with {
          \fill circle [radius=.28ex];
        }
      }
    }
  }
}

\DeclareSymbolFont{stmry}{U}{stmry}{m}{n}
\DeclareMathSymbol\fatsemi\mathop{stmry}{"23}

\DeclareFontFamily{U}{mathx}{\hyphenchar\font45}
\DeclareFontShape{U}{mathx}{m}{n}{
      <5> <6> <7> <8> <9> <10>
      <10.95> <12> <14.4> <17.28> <20.74> <24.88>
      mathx10
      }{}
\DeclareSymbolFont{mathx}{U}{mathx}{m}{n}
\DeclareFontSubstitution{U}{mathx}{m}{n}
\DeclareMathAccent{\widecheck}{0}{mathx}{"71}

\DeclareMathOperator{\Hom}{Hom}

\DeclareMathOperator*{\colim}{colim}

\DeclareMathOperator{\Ob}{Ob}

\newcommand{\cat}[1]{\mathcal{#1}}
\newcommand{\Cat}[1]{\mathbf{#1}}
\newcommand{\Fun}[1]{\mathrm{#1}}

\newcommand{\acset}{\mathrm{ACSet}}

\newcommand{\enriched}[1]{#1\text{-}\smcat}

\newcommand{\id}{\mathrm{id}}
\newcommand{\then}{\mathbin{\fatsemi}}
\newcommand{\Tr}[3]{\mathrm{Tr}^{#3}_{#1,#2}}

\newcommand{\To}[1]{\xrightarrow{#1}}

\newcommand{\from}{\leftarrow}
\newcommand{\From}[1]{\xleftarrow{#1}}

\newcommand{\inj}{\rightarrowtail}

\newcommand{\pto}{\rightharpoonup}

\newcommand{\tn}[1]{\textnormal{#1}}

\newcommand{\bb}{\mathbb{B}}
\newcommand{\nn}{\mathbb{N}}

\newcommand{\smset}{\Cat{Set}}
\newcommand{\smcat}{\Cat{Cat}}
\newcommand{\catsharp}{\smcat^\sharp}
\newcommand{\poly}{\Cat{Poly}}

\newcommand{\Maybe}{\Fun{Maybe}}
\newcommand{\List}{\Fun{List}}
\newcommand{\Writer}{\Fun{Writer}}
\newcommand{\Dist}{\Fun{Dist}}

\newcommand{\merge}[1]{\nabla_{#1}}

\newcommand{\qqand}{\qquad\text{and}\qquad}

\newcommand{\yon}{\mathcal{y}}
\newcommand{\tri}{\mathbin{\triangleleft}}

\newcommand{\cofree}[1]{\mathfrak{c}\langle#1\rangle}
\newcommand{\free}[1]{\mathfrak{m}\langle#1\rangle}

\newcommand{\meal}{\Cat{Meal}}
\newcommand{\mealquot}{\meal/\sim}
\newcommand{\dynmeal}{\Cat{DynMeal}}
\newcommand{\dk}[1][t]{\Cat{DK}_{#1}}

\newcommand{\trap}{\mathit{Trap}}
\newcommand{\iter}{\mathit{iter}}
\newcommand{\tripow}[1]{^{\tri #1}}

\newcommand*\circled[1]{\tikz[baseline=(char.base)]{
            \node[shape=circle,draw,solid,inner sep=2pt] (char) {#1};}}

\begin{document}
\title{Dynamic Tracing: a graphical language for rewriting protocols}

\author{ \href{https://orcid.org/0000-0002-9374-9138}{\includegraphics[scale=0.06]{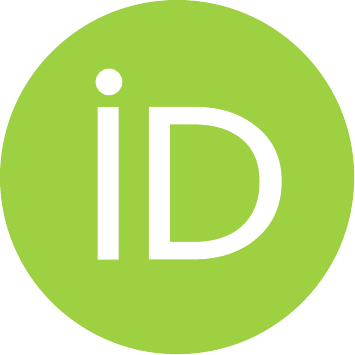}\hspace{1mm}Kristopher Brown} \\
	Topos Institute\\
	\texttt{kris@topos.institute} \\
	\And
	\href{https://orcid.org/0000-0002-9326-5328}{\includegraphics[scale=0.06]{orcid.pdf}\hspace{1mm}David I. Spivak} \\
	Topos Institute\\
	\texttt{david@topos.institute} 
} 
\date{}
\renewcommand{\headeright}{}
\renewcommand{\undertitle}{}

\maketitle 
\begin{abstract}
    The category $\smset_*$ of sets and partial functions is well-known to be traced monoidal, meaning that a partial function $S+U \pto T+U$ can be coherently transformed into a partial function $S \pto T$. This transformation is generally described in terms of an implicit procedure that must be run. We make this procedure explicit by enriching the traced category in $\catsharp$, the symmetric monoidal category of categories and cofunctors: each hom-category has such procedures as objects, and advancement through the procedures as arrows. We also generalize to traced Kleisli categories beyond $\smset_*$, providing a conjectural trace operator for the Kleisli category of any polynomial monad of the form $t+1$. The main motivation for this work is to give a formal and graphical syntax for performing sophisticated computations powered by graph rewriting, which is itself a graphical language for data transformation.

\end{abstract}

\section{Motivation}

Explicitly constructed programs are the standard means of specifying data transformation. However, by forgoing the full expressivity of a general programming language, one can work within a restricted syntax---a \emph{domain-specific language}---that has desirable properties. For example, flowcharts are often used as an informal syntax for software projects, where each box is associated with a subroutine that transforms data that flows on input wires into data that flows on output wires. Category theory makes this assignment of semantics precise: the syntax of directed wiring diagrams can formally be given a semantics in any symmetric monoidal category. There are advantages of this two-dimensional syntax over one-dimensional expression trees in a general purpose language, including transparent visualization, operadic substitution, and algebraic manipulation. Furthermore, when represented as a combinatorial object, the diagram itself is a particularly efficient normal form for a large number of syntax trees which it is equivalent to \cite{Patterson_2021}.    



The field of graph transformation uses the syntax of spans in an adhesive category, interpreting them as rewrite rules.
The semantics of deletion, copying, merging, and adding data can be attributed to these spans in various formalisms, notably DPO~\cite{ehrig1973graph}, SPO~\cite{lowe1993algebraic}, and SqPO~\cite{corradini2006sesqui} rewriting. If these sorts of operations are all one needs, it is advantageous to work with this syntax rather than general expression trees, as its rules are more easily visualized and subject to static analysis. 
 
Despite these virtues, there has been difficulty in applying graph transformation in engineering practice \cite{blostein1996issues,voss2023graph}. The cited reviews discuss various strategies for computation via graph rewriting. A straightforward method is unordered graph rewriting, where rewrite rules are applied in an arbitrary order, possibly with constraints; however, many applications require the expressivity of executing sequences of atomic rewrites in a systematic, domain-specific way, e.g.\ looping over a set of matches. Some earlier diagrammatic languages for this are based on directed graphs \cite{bunke1982attributed,fahmy1996reasoning}, where vertices are rewrite rules and edges are $\bb$-valued, indicating where to go next if the current rule either does match or doesn't match. There are also more expressive languages which propose a BNF grammar for graph programs\cite{schurr1991progress,plump2012design}. The control flow of such programs are implicit in the semantics given to constructors like \lstinline{while} rather than explicit (e.g.\ the looping of a diagram). Current abstractions for programming with graph rewriting make it difficult for engineers to collaborate due to incompatibilities between software implementations or between the ontologies presupposed by rewrite rules. Furthermore, fine-grained control over \emph{which} match is used in a rewrite is not emphasized.  

We will demonstrate how the structure of a rich class of data transformations can be given the structure of a symmetric monoidal category with a conjectural trace operator, which licenses the use of wiring diagrams with feedback loops as a graphical syntax. After describing this class abstractly, we demonstrate how this formalism guides the user interface and implementation of software designed to construct elaborate computations built up out of rewrite rules.

In \cref{chap.background} we lead up to the construction of a category $\dk$, parameterized by a polynomial monad~$t$, that will be sufficient for our rewriting application. In \cref{chap.dk_traced}, we show it is a cocartesian monoidal category enriched in $\catsharp$. We also offer a conjectural trace operator for monads meeting a certain criterion. In \cref{chap.application}, we produce a domain-specific language (DSL), with semantics in $\dk$, for constructing sophisticated programs that manipulate data via rewrite rules. We conclude with a summary and future work in \cref{chap.conclusion}.

\section*{Notation}

We assume familiarity with categories, functors, and enriched category theory. Many technical details related to $\poly$ and its many monoidal structures ($+,\times,\otimes,\tri$) can be referenced in \cite{poly}; we elide these to instead focus on providing intuitions for these constructions in Section \ref{chap.background}. In contrast, Section \ref{chap.dk_traced} requires a strong technical understanding of $\poly$. We will use the notation $[-,-]$ to refer to the internal hom in a category, while $(-,-)$ will be used for (co)pairing morphisms in a category with (co)products. We use $\sum$ to denote the disjoint union of sets. We often denote the identity on an object simply by the object name, e.g.\ we use $A$ to denote $\id_A$.

\section{Background}\label{chap.background}

In this section, we incrementally improve frameworks for modeling dynamical systems of the sort we need for rewriting by proposing a sequence of categories: $\meal$, $\mealquot$, $\dynmeal$, and finally, $\dk$. Each formalism should allow us to view a large system as the composition of smaller ones, where at any level of granularity we can consider local systems as enclosed boxes, which can be entered and exited and which evolve as we interact with them.

\subsection{Mealy machines}

 We first consider the category $\meal$ of Mealy machines:
\begin{align*}
    \Ob(\meal) &:= \Ob(\smset)\\
    \Hom_{\meal}(A,B) &:= \left\{\left(S : \smset,s_0 : S, 
    \phi\colon A\times S\to S\times B\right)\right\}
\end{align*}

A map could be called a \emph{dynamic function}. It includes a set $S$ of states and a particular state $s_0$. Further, given any state $s:S$, it provides both a function $A\to B$ and, for any input $a:A$, an updated state; all this is encoded in $\phi(-,s)\colon A\to S\times B$. Thus we can think of a morphism in $\meal$ as a function $A\to B$ that is updated that every time it receives an input. An example morphism in a toy model of chess pieces evading traps is given in Figure \ref{fig:Mealy}a-b.

\begin{figure}[h!]
    \centering
    \includegraphics[width=\textwidth]{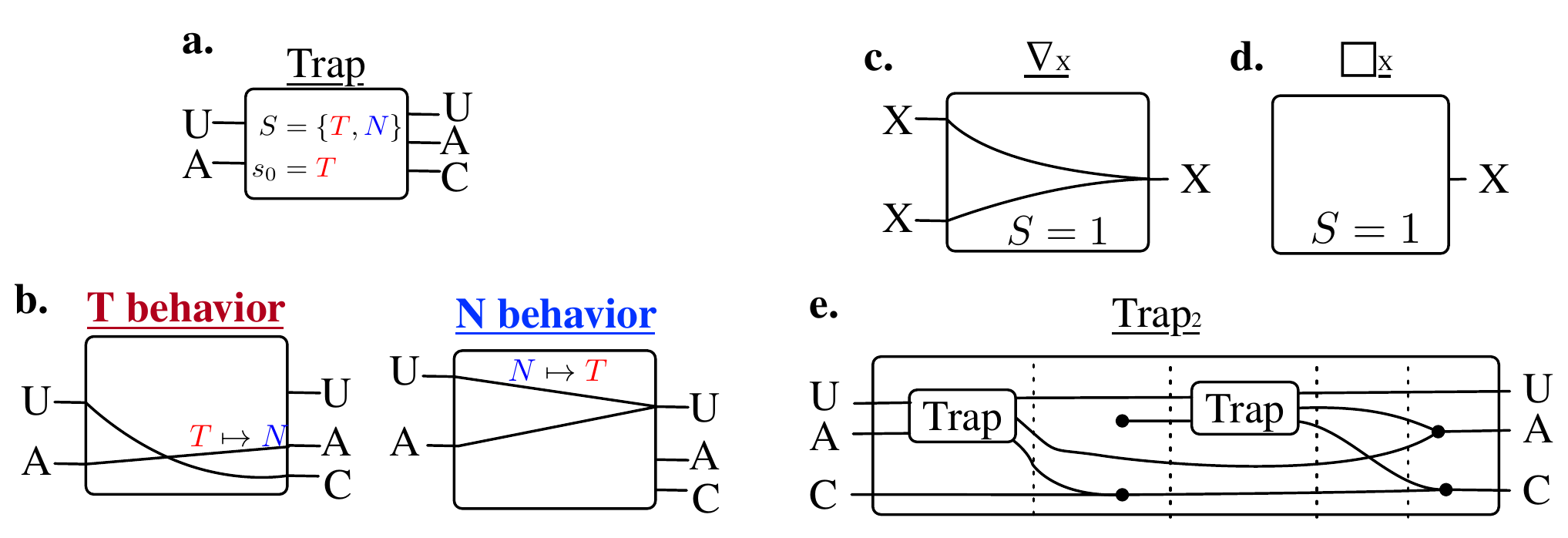}
    \caption{Suppose we play a game like chess where certain pieces may be in various conditions: \textbf{U}naware, \textbf{A}lert, and \textbf{C}aptured. \textbf{a.} A morphism $\trap : \Hom_{\meal}(\{U,A\},\{U,A,C\})$ represents a system where one's opponent may have set a trap to capture a piece: this is a dynamical system that accepts a piece (Unaware or Alert) and outputs an Unaware, Alert, or Captured piece. $\trap$ has two possible states: either there is a trap (\textcolor{red}{$T$}) or no trap(\textcolor{blue}{$N$}) set, with the system starting in state \textcolor{red}{$T$}. \textbf{b.}~$\trap$ requires a function $S\times \{U, A\}\to S\times \{U, A, C\}$, which is visualized by showing the dynamics in both possible states. In state \textcolor{red}{$T$}, unaware pieces fall for the trap and are captured, whereas alert pieces stay alert (and the state changes to \textcolor{blue}{$N$}). In state \textcolor{blue}{$N$}, all inputs exit as unaware, though a $U$ input triggers the state to change to \textcolor{red}{$T$}.  \textbf{c.}~A family of morphisms in $\meal$ of the form $\Hom_{\meal}(X+X,X)$, for any set $X$. We will visually depict this as a dot with multiple incoming wires. \textbf{d.} A family of morphisms in $\meal$ of the form $\Hom_{\meal}(\varnothing,X)$, for any set $X$. We will visually depict this as a dot with one outgoing wire. \textbf{e.} A composite dynamical system $\trap_2$ depicting pieces passing through two $\trap$ systems: Alert pieces exiting the first subsystem skip the second one. This diagram formally depicts the sequential composition of five morphisms, which are distinguished by the dotted vertical lines: $\trap\then(U+\square_A+ A + \merge C)\then(\trap + A+ C)\then(U+ A+ \sigma_{C,A}+ C)\then(U+ \merge A+\triangledown C)$. The state space of the composite system is $\{\textcolor{red}{T},\textcolor{blue}{N}\}\times \{\textcolor{red}{T},\textcolor{blue}{N}\}$, with initial state $(\textcolor{red}{T},\textcolor{red}{T})$.}
    \label{fig:Mealy}
\end{figure}

We claim $\meal$ is a symmetric monoidal category, with monoidal unit $\varnothing$ and the object $A \otimes B$ given by disjoint union of sets. For tensor and composition of morphisms, we define%
\footnote{
We are elide the symmetry $B\times S\cong S\times B$ in our notation `$\phi\times\psi$`.
}
\begin{align*}
(S,s_0,\phi) \otimes (T,t_0,\psi)&\coloneqq(S\times T, (s_0, t_0), \phi\times \psi)\\ (S,s_0,\phi) \then (T,t_0,\psi)&\coloneqq(S\times T, (s_0, t_0), (\phi\times T)\then(S \times \psi))
\end{align*}

Because $\meal$ is a symmetric monoidal category, there is a ready-made graphical language for visualizing serial and parallel compositions of its morphisms. The icons of this language are wires, each labeled with a set, boxes with input and output ports, also each labeled with a set, and braidings; these are described in detail in \cite[\S3.1, 3.5]{Selinger_2010}. In brief, data flows from left to right on wires, which are objects. The vertical dimension represents the tensoring $\otimes$ of objects and morphisms. We visualize an element of $\Hom_{\meal}(A,B)$ as a box with $|A|$ input ports and $|B|$ output ports. A diagrammatic example of a nontrivial composition of morphisms is provided in Figure \ref{fig:Mealy}e. 

We are ultimately working towards a traced monoidal category with coproducts; however, coproducts do not generally exist in $\meal$. 
The natural choice for a copairing of morphisms $(S,s_0,\phi) : \Hom_{\meal}(A,Z)$ and $(S',s'_0,\psi): \Hom_{\meal}(B,Z)$ is to place the boxes in parallel and merge their outputs, i.e.\ $(S,s_0,\phi) \otimes (S,s_0,\phi) \then \merge Z$, where $\merge{Z}\colon Z+Z\to Z$ is the codiagonal. This would have the correct \textit{behavior}, in a sense that will be made precise in the following section. However, the required coproduct equalities (e.g.\ $\iota_1\then(\phi,\psi) = \phi$) do not hold in general because the state spaces are not equal, i.e.\ $S\times S' \cong^? S$.

\subsection{Quotienting by behavioral equivalence}

 We now talk about behaviors as proper objects of study, rather than merely induced by Mealy machines. Although the set of behaviors of an Mealy machine with input $A$ and output $B$ can be succinctly characterized as the underlying set of the final coalgebra on $X\mapsto B^AX^A$, a set we will eventually denote by $\cofree{[A\yon,B\yon]}(1)$, we will more naturally be able to generalize our results by viewing behavior through the lens of polynomial functors; indeed, this is where the above notation comes from: $[A\yon,B\yon]$ is an internal hom in $\poly$ and $\cofree{-}$ is the cofree comonoid construction on $\poly$. 
 
 Polynomial functors are formally sums of representable functors $\smset\to\smset$, i.e.\ functors of the form $\yon^A := \Hom_\smset(A,-): \smset \to \smset$, sending any set $X$ to $X^A$. The sum is indexed by a set $I$ such that the polynomial can be denoted $p := \sum_{i : I} \yon^{A_i}$. Note that $I$ is canonically isomorphic to $p(1)$. The elements of $p(1)$ are called positions, while for each position $i: p(1)$ there is a set of directions $A_i$. We denote $A_i$ by $p[i]$, so we may write $p\cong\sum_{i:p(1)}\yon^{p[i]}$. Polynomial functors can be thought of as system interfaces, and a few ways to represent them are shown below.

\begin{center}
\[
\begin{array}{ccccccc}
\textbf{Interface}&&\textbf{Algebraic}&&\textbf{Bundle}&&\textbf{Corolla forest}\\
\begin{array}{l}
\{\text{\textcolor{red}{Listen}}\}\yon^{\{L,R\}} \\ 
+\ \text{\{\textcolor{dgreen}{Left}, \textcolor{blue}{Right}\}}\yon^{\text{\{Move\}}} \\ 
+\ \text{\{\textcolor{purple}{Stop}\}}\yon^{\{\}}
\end{array}
&&\yon^2+2\yon+1&&
\begin{tikzpicture}[baseline=(pi)]
  \foreach \n/\i/\c in {1/2/red,2/1/dgreen,3/1/blue,4/0/purple} {
  	\node at (.5*\n,0) (b\n) {\color{\c}$\bullet$};
		\node[] (b\n0) at (.5*\n,.8) {};
  	\ifthenelse{\i>0}{
		\foreach \j in {1,...,\i} {
		\node (b\n\j) at (.5*\n,.5+.3*\j) {$\bullet$};
		};
		}{};
  };
  \node[draw, inner sep=1pt, rounded corners=5, fit=(b1) (b4)] (B)  {};
  \node[draw, inner sep=1pt, rounded corners=5, fit=(b12) (b40)] (E) {};
  \draw[->, thick] (E) -- node[left, font=\tiny] (pi) {$\pi$} (B.north-|E);
\end{tikzpicture}
&&
\begin{tikzpicture}[trees, out distance=-2pt]
  \node (1) {$\textcolor{red}{\bullet}$} 
    child {}
    child {};
  \node[right=.5 of 1] (2) {$\textcolor{dgreen}{\bullet}$} 
    child {};
  \node[right=.5 of 2] (3) {$\textcolor{blue}{\bullet}$} 
    child {};
  \node[right=.5 of 3] (4) {$\textcolor{purple}{\bullet}$};
\end{tikzpicture}
\end{array}
\]   
    \captionof{figure}{Various representations of a polynomial functor that characterizes an interface which can do three things: it can listen (and receive a boolean value), it can move (left or right), or it can stop.}
    \label{fig:poly}

\end{center}

Polynomial comonads can be identified with categories \cite{Ahman_2016} and cofunctors, a category we denote $\catsharp$. The cofree comonad functor $\cofree{-}\colon\poly\to\catsharp$ can be thought of as sending a polynomial $p$ to a category whose objects are possible behavior trees of a system with interface $p$. These are potentially infinite trees and are obtained by starting with a root and stacking one-step behaviors, as seen in the corolla forest representation of the polynomial. To provide intuition for what these behavior trees look like, Figure~\ref{fig:morph} shows a tree corresponding to the Mealy machine of Figure~\ref{fig:Mealy}a-b as well as a tree for the above example interface. This allows us to define the category, $\mealquot$:

\begin{equation}\label{mealquot}
\begin{aligned}
\Ob(\mealquot) &:= \Ob(\smset)\\
\Hom_{\mealquot}(A,B) &:= \Ob\cofree{[A\yon,B\yon]}
\end{aligned}
\end{equation}

\begin{figure}[h!]
    \centering
    \begin{minipage}{.5\textwidth}
\[
\begin{tikzpicture}[trees, scale=1.5, level distance=20pt, out distance=1pt,
  level 1/.style={sibling distance=15mm},
  level 2/.style={sibling distance=5mm},
  level 3/.style={sibling distance=2.5mm},
  level 4/.style={sibling distance=1.25mm},
  level 5/.style={sibling distance=1.25mm}]
  \node[red] (a) {\large $\textcircled{T}$}
    child {node[red] {\large $\textcircled{T}$}
        child {node[red] {\large $\textcircled{T}$}
            child { }
            child { }
        edge from parent node[left] {$U$\ \ }
        }
        child {node[blue] {\large $\textcircled{N}$}
            child { }
            child { }
        edge from parent node[right] {\ $A$}
        }
    edge from parent node[left] {$U$\ \ }
    }
    child {node[blue] {\large $\textcircled{N}$}
        child {node[red] {\large $\textcircled{T}$}
            child { }
            child { }
        edge from parent node[left] {$U$\ \ }
        }
        child {node[blue] {\large $\textcircled{N}$}
            child { }
            child { }
        edge from parent node[right] {\ $A$}
        }
    edge from parent node[right] {\ $A$}
    }
  ;
\end{tikzpicture}
\]
\end{minipage}\begin{minipage}{.4\textwidth}
\[
\begin{tikzpicture}[trees, scale=1.5, out distance=.5pt,
  level 1/.style={sibling distance=17mm},
  level 2/.style={sibling distance=5mm},
  level 3/.style={sibling distance=2.5mm},
  level 4/.style={sibling distance=1.25mm},
  level 5/.style={sibling distance=1.25mm}]
  \node[red] (a) {$\bullet$}
        child {node[dgreen] {$\bullet$}
            child {node[red] {$\bullet$}
                        child {node[dgreen] {$\bullet$}
                child {node[purple] {$\bullet$}
                edge from parent node[left] {$\mathit{Move}$}
                }
            edge from parent node[left] {$L$\ }
            }
            child {node[blue] {$\bullet$}
                child {node[purple] {$\bullet$}
                edge from parent node[right] {$\mathit{Move}$}
                }
            edge from parent node[right] {\ $R$}
            }
            edge from parent node[left] {$\mathit{Move}$}
            }
        edge from parent node[left] {$L$\ }
        }
    child {node[blue] {$\bullet$}
        child {node[red] {$\bullet$}
            child {node[dgreen] {$\bullet$}
                child {node[purple] {$\bullet$}
                edge from parent node[left] {$\mathit{Move}$}
                }
            edge from parent node[left] {$L$\ }
            }
            child {node[blue] {$\bullet$}
                child {node[purple] {$\bullet$}
                edge from parent node[right] {$\mathit{Move}$}
                }
            edge from parent node[right] {\ $R$}
            }
        edge from parent node[right] {$\mathit{Move}$}
    }
    edge from parent node[right] {\ $R$}
    }
  ;
\end{tikzpicture}
\]
\end{minipage}
    \caption{\textbf{Left.} An object of $\cofree{[\{U, A\}\yon, \{U, A, C\}\yon]}$ which is the behavior of $\trap$ from Figure \ref{fig:Mealy}b. Note that $[2\yon,3\yon]\cong 3^2\yon^2$, i.e.\ it is an interface that receives one of three possible inputs and can be configured in any possible function $2\to 3$. We label two of these eight possible functions as \textcolor{red}{$T$} and \textcolor{blue}{$N$}. \textbf{Right.} One possible behavior for the interface of Figure \ref{fig:poly}, i.e.\ an object of the category $\cofree{\yon^2+2\yon+1}$. The machine starts in listening state, and upon receiving the input $L$ (respectively $R$) it moves left (resp. right). The machine repeats this once more and then stops.}
    \label{fig:morph}
\end{figure}

An important feature is not explicitly modeled. We have $\Hom(A,B)$ being a mere \textit{set} of behaviors, but we need to make use of a rich structure that this set of behaviors has when evolving the system over time: the way the system changes as new inputs are received. Enriching our category, i.e.\ replacing each set of morphisms by a category, will allow us to characterize how the morphisms change as inputs are received.

\subsection{Modeling the system evolution in time}

Recall that $\cofree{p}$ sends a polynomial $p$ to a category. Until now we have only considered the objects of that category. The morphism $\mathit{tree}_1\to \mathit{tree}_2$ is given by a path in $\mathit{tree}_1$, starting from the root and ending at a copy of $\mathit{tree}_2$ (see Figure \ref{fig:comon}). This is crucial for explicitly modeling how a sequence of inputs leads to a new dynamical system. We define $\dynmeal$ to have the same objects as $\meal$ but to be enriched in $\catsharp$ rather than $\smset$, such that the hom-object $\Hom_{\dynmeal}(A,B):=\cofree{[Ay,By]}$.

\begin{figure}[h!]
    \centering

\begin{minipage}{.4\textwidth}
\[
\begin{tikzpicture}[trees, scale=1.5,
  level 1/.style={sibling distance=20mm},
  level 2/.style={sibling distance=10mm},
  level 3/.style={sibling distance=5mm},
  level 4/.style={sibling distance=2.5mm},
  level 5/.style={sibling distance=1.25mm}]
  \node[dgreen] (a) {$\bullet$}
    child[ultra thick] {[solid] node[red] {$\bullet$}
        child {node[red] {$\bullet$}
            child[thin] {[solid] node[purple] {$\bullet$}
            edge from parent node[left] {$L$\quad }
            }
            child[thin] {[solid] node[blue] {$\bullet$}
                child {node[dgreen] {$\bullet$}
                    child {node[purple] {$\bullet$}
                edge from parent node[right] {$\mathit{Move}$}
                }
            edge from parent node[right] {$\mathit{Move}$}
                }
            edge from parent node[right] {$R$}
            }
        edge from parent node[left] {$L$\ \ \ }
        }
        child {[thin] node[blue] {$\bullet$}
            child {node[purple] {$\bullet$}
            edge from parent node[right] {$\mathit{Move}$}
            }
        edge from parent node[right] {\ $R$}
        }
    edge from parent node[left] {$\mathit{Move}$\ \ }
    } 
  ;
\end{tikzpicture}
\]
\end{minipage}\begin{minipage}{.1\textwidth}\mbox{\normalsize$\overset{\mathit{Move};L}\longrightarrow$}\end{minipage}\begin{minipage}{.4\textwidth}
\[
\begin{tikzpicture}[trees, scale=1.5,
  level 1/.style={sibling distance=5mm},
  level 2/.style={sibling distance=2.5mm},
  level 3/.style={sibling distance=1.25mm},
  level 4/.style={sibling distance=2.5mm},
  level 5/.style={sibling distance=1.25mm}]
  \node[red] (a) {$\bullet$}
    child {node[purple] {$\bullet$}
    edge from parent node[left] {$L$\ }
    }
    child {node[blue] {$\bullet$}
        child {node[dgreen] {$\bullet$}
            child {node[purple] {$\bullet$}
        edge from parent node[right] {$\mathit{Move}$}
        }
    edge from parent node[right] {$\mathit{Move}$}
    }
    edge from parent node[right] {\ $R$}
    }
  ;
\end{tikzpicture}
\]
\end{minipage}

    \caption{A morphism between two objects in $\cofree{\yon^2+2\yon+1}$, with the data of the morphism represented by thick arrows, i.e. $\mathit{Move}$, then $L$.}
    \label{fig:comon}
\end{figure}

The category $\catsharp$ is that of $\tri$-comonoids in $\poly$, which is equivalent to the category of categories and \emph{co}functors as proven in \cite{Ahman_2016,niupolynomial}. A cofunctor has a map in the forward direction on objects but a map on morphisms in the reverse direction. This means, for a composite system $\phi \cdot \psi$, we can construct a composite behavior tree given behavior trees for $\phi$ and $\psi$, but the way  $\phi$ and $\psi$ evolve over time is dictated by how $\phi \cdot \psi$ evolves over time, for example, see Figure \ref{fig:behaviormorph}. Every comonoid in $\poly$ has an underlying polynomial, and the associated functor $U$ has a right adjoint $\cofree{-}$:
\begin{equation}\label{eqn.cofree_adjunction}
\begin{tikzcd}
    \poly\ar[r, shift left=7pt, "\cofree{-}"]\ar[r, phantom, "\top"]&
    \catsharp\ar[l, shift left=7pt, "U"]
\end{tikzcd}
\end{equation}

\begin{figure}[h!]
    \centering
    
\begin{minipage}{.4\textwidth}
\[
\begin{tikzpicture}[trees, scale=2,level distance=20pt, out distance=2pt, 
  level 1/.style={sibling distance=10mm},
  level 2/.style={sibling distance=2.5mm},
  level 3/.style={sibling distance=5mm},
  level 4/.style={sibling distance=2.5mm},
  level 5/.style={sibling distance=1.25mm}]
  \node (a) {\circled{${\textcolor{red}{TT}}$}}
    child { node {\circled{$\textcolor{red}{TT}$}}
        child {[solid]}
        child {[solid]}
        child {[solid]}
    edge from parent node[left] {$U$\ \ \ }
    }
    child[ultra thick] {node {\circled{${\textcolor{blue}{N}\textcolor{red}{T}}$}}
        child {[thin]}
        child {[thin]}
        child {[thin]}
    edge from parent node[left] {$A$\ \ }
    }
    child { node {$\circled{\textcolor{red}{T}\textcolor{red}{T}}$}
        child {[solid]}
        child {[solid]}
        child {[solid]}
    edge from parent node[right] {\ $C$}
    }
  ;
\end{tikzpicture}
\]
\end{minipage}\begin{minipage}{.1\textwidth}\mbox{\Huge$\Longrightarrow$}\end{minipage}\begin{minipage}{.3\textwidth}
\[
\begin{tikzpicture}[trees, scale=1.5, level distance=20pt, out distance=2pt,
  level 1/.style={sibling distance=10mm},
  level 2/.style={sibling distance=2.5mm},
  level 3/.style={sibling distance=1.25mm},
  level 4/.style={sibling distance=2.5mm},
  level 5/.style={sibling distance=1.25mm}]
  \node[red] (a) {$\circled{T}$}
    child { node[red] {$\circled{T}$}
        child {[thin] }
        child {[thin] }
    edge from parent node[left] {$U$\ \ }
    }
    child[ultra thick] {node[blue] {$\circled{N}$}
        child {[thin] }
        child {[thin] }
    edge from parent node[right] {\ \ $A$}
    }
  ;
\end{tikzpicture}
\]
\end{minipage}\begin{minipage}{.3\textwidth}
\[
\begin{tikzpicture}[trees, scale=1.5,level distance=20pt, out distance=2pt,
  level 1/.style={sibling distance=10mm},
  level 2/.style={sibling distance=2.5mm},
  level 3/.style={sibling distance=1.25mm},
  level 4/.style={sibling distance=2.5mm},
  level 5/.style={sibling distance=1.25mm}]
  \node[red] (a) {$\circled{T}$}
    child {node[red] {$\circled{T}$}
        child {[solid] }
        child {[solid] }
    edge from parent node[left] {$U$\ \ }
    }
    child { node[blue] {$\circled{N}$}
        child {[solid] }
        child {[solid] }
    edge from parent node[right] {$A$}
    }
  ;
\end{tikzpicture}
\]
\end{minipage}

\begin{minipage}{.4\textwidth}
$:\Hom \cofree{[\{U, A, C\} \yon, \{U, A, C\} \yon]}$
\end{minipage}\begin{minipage}{.2\textwidth}
\ 
\end{minipage}\begin{minipage}{.4\textwidth}
$:\Hom \cofree{[\{U, A\}\yon, \{U, A, C\}\yon]}$
\end{minipage}

    \caption{Example of the change in behavior of the composite system $\trap_2$ (Figure~\ref{fig:Mealy}e) in response to an input  $A$. Just like in Figure~\ref{fig:morph}, where $\textcolor{red}{T}$ and $\textcolor{blue}{N}$ represented elements of $3^2$, the labels $\textcolor{red}{TT}$ and $\textcolor{blue}{N}\textcolor{red}{T}$ refer to elements of $3^3$, induced by where $\trap_2$ sends its inputs  data of the wiring diagram of $\trap_2$ instructs how to convert this $A$-interaction into interactions for its $\trap$ subcomponents. Note that the second $\trap$ subcomponent is not updated at all for an $A$ input, i.e. its update is an identity morphism in $\cofree{[\{U, A\}\yon, \{U, A, C\}\yon]}$.}
    \label{fig:behaviormorph}
\end{figure}

There is a strong connection between this story and that of coalgebras. For any polynomial $p$, the set of $\cofree{p}$ objects is isomorphic to the underlying set of the final coalgebra on $X\mapsto B^AX^A$. The category of functors $\cofree{p}\to\smset$ is equivalent to the category of $p$-coalgebras. Each behavior tree in $\cofree{p}$ is sent to the set of states with that behavior. 

The enriched structure of $\dynmeal$ now makes explicit how our behavior trees change in response to inputs. The last improvement to be made is one of expressivity, e.g.\ adding the possibility of entering an environment and failing to ever exit it, considering lists of possible outcomes, and considering probability distributions of possible outcomes.

\subsection{Adding monadic effects}

The expressivity captured so far can be vastly generalized by incorporating a polynomial monad $(t,\eta,\mu)$ on $\smset$ into our morphisms. Particular monads of interest are

\begin{equation}\label{eqn.monads_of_interest}
\Maybe= \yon+1,\qquad
\Writer_M= M\yon, \qquad
\List= \sum_{N:\nn}\yon^N, \qquad
\text{and }\Dist= \sum_{N:\nn}\Delta_N\yon^N
\end{equation}
where $M$ is a monoid and ${\Delta_N=\{P:N \to [0,1]\ |\ 1=\sum P(i)\}}$. Our work so far has been general to polynomials, not merely interfaces of the form $Ay\to By$, so monadic effects can be incorporated into our category by considering the Hom object of an interface $A\to B$ to be $\cofree{[Ay,t \tri By]}$.

\begin{definition}[Dynamic Kleisli category]\label{def.dynamic_kleisli}

Given a polynomial monad $t$, we define a category enriched in $\catsharp$, denoted $\dk$, as follows:

\begin{align*}
\Ob(\dk) &:= \Ob(\smset)\\
\Hom_{\dk}(A,B) &:= \cofree{[Ay, t \tri By]}
\end{align*}
\end{definition}

\section{$\dk$ as a traced monoidal category}\label{chap.dk_traced}

Throughout this section we assume that $(t,\eta,\mu)$ is a Cartesian polynomial monad. This is sufficient to show that $\dk$ is a cocartesian monoidal $\catsharp$-enriched category. In \cref{sec.dk_cat} we will show it is a $\catsharp$-enriched category, and in \cref{sec.dk_coprod} we will show it has coproducts.

In fact, we would like $\dk$ to be traced monoidal, meaning that there are morphisms
\begin{equation}\label{eqn.tracemap}
\Tr{A}{B}{U}\colon\Hom_{\dk}(A\otimes U,B\otimes U)\to\Hom_{\dk}(A,B)
\end{equation}
satisfying various compatibility conditions. Traced monoidal categories have a graphical syntax that includes loops
\[
\begin{tikzpicture}[oriented WD,baseline=(cod.center), bbx=1em, bby=1ex]
    \node[bb={2}{2}, bb name=$f$] (dom) {};
    \node[bb={0}{0}, fit={(dom) ($(dom.north east)+(1,6)$) ($(dom.south west)-(1,1)$)}, bb name = $\tn{Tr}(f)$] (cod) {};
    \coordinate (cod_in1) at (dom_in2-|cod.west);
    \coordinate (cod_out1) at (dom_out2-|cod.east);
	\draw[ar,pos=20] (cod_in1) to (dom_in2);
	\draw[ar,pos=2] (dom_out2) to (cod_out1);
	\draw[ar] let \p1=(dom.north east), \p2=(dom.north west), \n1={\y2+\bby}, \n2=\bbportlen in (dom_out1) to[in=0] (\x1+\n2,\n1) to[out=180,in=0] node[above=2pt, label] {$U$} (\x2-\n2,\n1) to[out=180] (dom_in1);
	\draw[label] 
		node[below left=2pt and 3pt of dom_in2]{$A$}
		node[below right=2pt and 3pt of dom_out2]{$B$};
\end{tikzpicture}
\]
(see also \cite{Selinger_2010}). This notion was defined for $\smset$-enriched categories in \cite{joyal1996traced}, and the definition can be extended to the $V$-enriched setting by asking that the trace map \eqref{eqn.tracemap} be a map in $V$. In Section \ref{sec.dk_traced}, we will propose a trace map for $\dk$, whenever $t$ is exceptional in the sense of Definition \ref{def.exceptional}.

\subsection{$\catsharp$-enriched category structure on $\dk$}\label{sec.dk_cat}
In this section we suppose a good deal more familiarity with $\poly$. The proposed $\catsharp$-enriched category structure on $\dk$ was given in Definition \ref{def.dynamic_kleisli}; our first goal in this section is to prove that it satisfies the correct properties. In \cref{sec.dk_coprod} we will show that it is monoidal.

For any lax monoidal functor $\cat{V}\to\cat{W}$, there is an induced functor $\enriched{\cat{V}}\to\enriched{\cat{W}}$. For all monoidal categories $(\cat{V},I,\otimes)$, the functor $\cat{V}(I,-)\colon\cat{V}\to\smset$ is lax monoidal, and we call the induced functor $\enriched{V}\to\smcat$ as the underying category functor. We say that a category $\cat{C}$ is enriched in $\cat{V}$ when there is a $\cat{V}$-category for which $\cat{C}$ is the underlying category. Our second goal in this section is to show that the usual Kleisli category $\smset_t$ is enriched in $(\catsharp,\yon,\otimes)$.

Our strategy for the first goal is to show that $\smset_t$ is enriched in $(\poly,\yon,\otimes)$.

\begin{theorem}
Let $(t,\eta,\mu)$ be a polynomial monad. The Kleisli category $\smset_t$ is enriched in $(\poly,\yon,\otimes)$.
\end{theorem}
\begin{proof}
For sets $A,B:\smset$, define the polynomial
\begin{equation}\label{eqn.polyhom}
h^t_{A,B}\coloneqq[A\yon,t\tri B\yon].
\end{equation}
Maps of the form $\yon\to h^t_{A,B}$ are in bijection with polymomial maps $A\yon\to t\tri B\yon$. Note that these are in bijection with functions $A\to t\tri B$, which are exactly the Kleisli morphisms, i.e.\ we have
\begin{equation}\label{eqn.enriched}
\poly(\yon,h^t{A,B})\cong\smset_t(A,B).
\end{equation}
So it suffices to define an identity and a unital and associative composition law for the hom-objects $h^t_{A,B}:\poly$.

For the identity on $A$ we use
\[A\yon\cong\yon\tri A\yon\To{\eta\tri A\yon}t\tri A\yon.\]
Maps of the form $h^t_{A,B}\otimes h^t_{B,C}\to h^t_{A,C}$ are in bijection with maps $A\yon\otimes[A\yon,t\tri B\yon]\otimes[B\yon,t\tri C\yon]\to C\yon$, so beginning with the evaluation map $A\yon\tri[A\yon,t\tri B\yon]\to t\tri B\yon$, it suffices to find a map $(t\tri B\yon)\otimes[B\yon,t\tri C\yon]\to C\yon$. Since $\poly$ is duoidal, we have the desired map:
\[
    (t\tri B\yon)\otimes(\yon\tri[B\yon,t\tri C\yon])\To{\tn{duoid}} (t\otimes\yon)\tri(B\yon\otimes[B\yon,t\tri C\yon])\To{\tn{eval}} t\tri t\tri C\yon\To{\mu}t\tri C.
\]
It is easy to check that these definitions are associative and unital.
\end{proof}

\begin{corollary}
Let $(t,\eta,\mu)$ be a polynomial monad. The Kleisli category $\smset_t$ is enriched in $(\catsharp,\yon,\otimes)$.
\end{corollary}
\begin{proof}
    For sets $A,B$, let $h^t_{A,B}:\poly$ be as in \eqref{eqn.polyhom}. Since the functor $\cofree{-}\colon\poly\to\catsharp$ is lax monoidal, we can define hom-objects
    \[    \dk(A,B)\coloneqq\cofree{h^t_{A,B}}=\cofree{[A\yon,t\tri B\yon]}
    \]
    giving rise to a $\catsharp$-category. To see that its underlying category is $\smset_t$, we use \cref{eqn.enriched,eqn.cofree_adjunction}
    \[  \catsharp(\yon,\cofree{h^t_{A,B}})\cong\poly(\yon,h^t_{A,B})\cong\smset_t(A,B),
    \]
    completing the proof.
\end{proof}

As a $\catsharp$-category, $\dk$ not only has an underlying ordinary category, induced by the lax monoidal functor $\catsharp(\yon,-)\colon\catsharp\to\smset$, but another corresponding ordinary category as well, induced by the lax monoidal functor $\Ob\colon\catsharp\to\smset$. There is a natural transformation between these two:
\[
\begin{tikzcd}[column sep=50pt]
    (\mathbf{\catsharp},\yon,\otimes)
        \ar[r, bend left=10pt, shift left=3pt, "{\catsharp(\yon,-)}", ""' name=top]
        \ar[r, bend right=10pt, shift right=3pt, "\Ob"', "" name=bot]&
    (\smset,1,\times)
    \ar[from=top, to=bot, Rightarrow]
\end{tikzcd}
\]
By Definition \eqref{mealquot}, we recover $\mealquot$ from $\dk$ by locally applying $\Ob$, as in the following.

\begin{proposition}
The functor $\enriched{\catsharp}\to\enriched{\smset}$ induced by $\Ob\colon\catsharp\to\smset$ sends $\dk\mapsto\mealquot$.
\end{proposition}

\subsection{Coproduct monoidal structure on $\dk$}\label{sec.dk_coprod}

We next show that the $\catsharp$-category $\dk$ has coproducts, as defined in \cite[\S3.8]{kelly1982basic}
and explicated in \cite{4354063}.

We need to show that for any two sets $A,B:\smset$, there are morphisms
\[
i\colon\yon\to\cofree{h^t_{A,A+B}}
\qqand
j\colon\yon\to\cofree{h^t_{B,A+B}}
\]
such that for each $C:\smset$, each $v:\catsharp$ and each pair of maps $f\colon v\to\cofree{h^t_{A,C}}$ and $g\colon v\to\cofree{h^t_{B,C}}$, there is a unique morphism $(f,g)\colon v\to\cofree{h^t_{A+B,C}}$ such that the following diagrams commute
\begin{equation}\label{eqn.coprod}
\begin{tikzcd}[column sep=35pt]
    \yon\otimes v\ar[d, "\cong"']\ar[r, "{i\otimes (f,g)}"]&
\cofree{h^t_{A,A+B}}\otimes\cofree{h^t_{A+B,C}}\ar[d, "\then"]\\
    v\ar[r, "f"']&\cofree{h^t_{A,C}}
\end{tikzcd}
\hspace{.5in}
\begin{tikzcd}[column sep=35pt]
    \yon\otimes v\ar[d, "\cong"']\ar[r, "{j\otimes (f,g)}"]&
\cofree{h^t_{B,A+B}}\otimes\cofree{h^t_{A+B,C}}\ar[d, "\then"]\\
    v\ar[r, "g"']&\cofree{h^t_{B,C}}
\end{tikzcd}
\end{equation}

\begin{theorem}
    For any polynomial monad $(t,\eta,\mu)$, the $\catsharp$-category $\dk$ has coproducts.
\end{theorem}
\begin{proof}
    For any $p,q,x:\poly$, the inclusions $p\to p+q\from q$ induce an isomorphism
    \[
    [p+q,x]\To{\cong}[p,x]\times [q,x]
    \]
    in $\poly$, so in particular we have
    \[h^t_{(A+B),C}\cong h^t_{A,C}\times h^t_{B,C}.\]
    By \eqref{eqn.cofree_adjunction}, maps $v\to\cofree{p}$ in $\catsharp$ are in bijection with maps $v\to p$ of polynomials. Thus we can identify a pair of $\catsharp$-maps $f\colon v\to\cofree{h^t_{A,C}}$ and $g\colon v\to\cofree{h^t_{B,C}}$ with a single $\poly$-map $(f,g)\colon v\to h^t_{A+B,C}$. The commutativity of the diagrams from \eqref{eqn.coprod} follows easily.
\end{proof}

\subsection{Proposed traced structure on $\dk$}\label{sec.dk_traced}

In this section we propose a trace structure on $\dk$ for certain polynomial monads $t$, which we call \emph{exceptional} because they are equipped with an element $\xi\colon 1\to t$ that acts like an exception: if any branches of a syntax tree throw an exception then so does the whole tree. 

Defining the trace map requires even more background on $\poly$ than the previous sections. Whereas the $\catsharp$-category structure on $\dk$, as defined in \cref{sec.dk_cat,sec.dk_coprod} was induced by a simpler, $\poly$-category structure on $\dk$, the application to rewriting and the trace structure both make use of the full $\catsharp$-enrichment.

Our first goal is to define exceptional monads.

\begin{lemma}\label{lemma.exceptional}
    For any cartesian polynomial monad $(t,\eta,\mu)$, the polynomial $t+1$ also carries the structure of a polynomial monad, and the coproduct inclusion $t\to t+1$ is a morphism of monads.
\end{lemma}
\begin{proof}
    It suffices to show that for any $t$ there is a distributive law $t\tri(\yon+1)\to(\yon+1)\tri t$ that commutes with the inclusion $\yon\to\yon+1$ on both sides.
    
    In any category with a terminal object, a coproduct inclusion $A\inj B$, i.e.\ an isomorphism $A+A'\cong B$ for some $A'$, induces a map
    \[
    B\To{\cong} A+A'\To{A+!} A+1
    \]
    such that $A\to B\to A+1$ is the coproduct inclusion. 

    It is easy to show that every cartesian monomorphism in $\poly$ is a coproduct inclusion. Moreover, for any polynomial $p$, the map $p\cong p\tri\yon\to p\tri(\yon+1)$ is a cartesian monomorphism because $\yon\to\yon+1$ is, and $\tri$ preserves monomorphisms and cartesian maps in both variables. Thus we have a map $p\tri(\yon+1)\to p+1\cong(\yon+1)\tri p$, natural with respect to cartesian maps $p\to p'$. One can check that when $t$ is a cartesian monad, this map is always a distributive law, completing the proof.
\end{proof}

\begin{definition}[Exceptional monad]\label{def.exceptional}
    Let $(t,\eta,\mu)$ be a cartesian polynomial monad. An \emph{exception} structure on $t$ is a retraction $\xi\colon t+1\to t$ of the monad inclusion $t\to t+1$ from Lemma \ref{lemma.exceptional}.
\end{definition}

By Lemma \ref{lemma.exceptional}, $t+1$ is an exceptional monad for any cartesian polynomial monad $t$. 

\begin{remark}
In the multiplication $\mu\colon t\tri t\to t$ of an exceptional monad, a position of the left-hand side consists of a position $I:t(1)$ and, for every direction $i:t[I]$, a position $J_i:t(1)$. If either $I$ or any of the the $J_i$ is the exceptional element $\xi$, then its image under $\mu$ must be the exceptional element: an exception anywhere causes an exception in the whole computation.

The notion of exceptional monad differs from that of \emph{monad with zero} \cite{wisnesky2011minimizing}
even though in each case the monad $t$ is equipped with a constant $1\to t$. For example, the $\List$ monad has a zero, namely the empty list, but it is not exceptional because a list of lists can contain the empty list without its concatenation being empty. 
\end{remark}

We next propose our trace map $\Tr{A}{B}{U}\colon\cofree{[(A+U)\yon,t\tri(B+U)\yon]}\to\cofree{[A\yon,t\tri B\yon]}$. In fact, it is more straightforward to define an \emph{iteration} map as in \cite{Selinger_2010} of the form
\begin{equation}\label{eqn.iter}
\iter^A_B\colon\cofree{[A\yon,t\tri(A+B)\yon]}\to\cofree{[A\yon,t\tri B\yon]}
\end{equation}
at which point we can define $\Tr{A}{B}{U}$ to be given by the composite
\begin{align*}
    \cofree{[(A+U)\yon,t\tri(B+U)\yon]}&\to
    \cofree{[(A+U)\yon,t\tri(A+U+B)\yon]}\\&\To{\iter^A_B}
    \cofree{[(A+U)\yon,t\tri B\yon]}\\&\to
    \cofree{[A\yon,t\tri B\yon]}.
\end{align*}

To get there, we need to be more explicit about the cofree comonad construction and the free monad construction on a polynomial $p$. 

The free monad $\free{p}$ on a pointed polynomial $\yon\to p$ is constructed in two steps. For any finitary pointed polynomial $q$---one for which each $q[J]$ is a finite set---the free monad on $q$ can be constructed in one step, namely as the colimit:
\[
	\free{q}\coloneqq\colim(\cdots\from q\tripow{n+1}\From{g_n}q\tripow{n}\from\cdots\from q)
\]
where the maps $g_n\colon q\tripow{n}\to q\tripow{n+1}$ are defined inductively by $p\tripow{n}\cong\yon\tri p\tripow{n}\to p\tri p\tripow{n}=p\tripow{n+1}$.
An arbitrary polynomial $p$ can be written as the filtered limit of its vertical projections $p\to p^{j}$ onto finitary polynomials: that is, for each sum-component $I: p(1)$, just replace $p[I]$ by an arbitrary finite subset of it, and take the limit of all such things under component-wise projection. That limit is isomorphic to $p$, and we write $p\cong\lim_{j: J_p}p^{(j)}$. By construction, each of these $p^{(j)}$ is finitary, so let $\free{p^{(j)}}\coloneqq\free{p^{(j)}}$ denote the free monad on it, constructed as above. Then finally we construct the free monad $\free{p}$ on $p$ as their filtered limit:
\[
	\free{p}\coloneqq\lim_{j: J_p}\free{p^{(j)}}.
\]

The cofree comonoid $\cofree{p}$ on $p:\poly$ is constructed in just one step. It is carried by the limit
\[
\cofree{p}\coloneqq\lim(\cdots\to p^{(n+1)}\To{f^{(n)}} p^{(n)}\to\cdots\to p^{(1)}\To{f^{(0)}} p^{(0)})
\]
where the $p^{(k)}$ are defined inductively as follows:
\begin{align*}
	p^{(0)}&\coloneqq\yon&p^{(k+1)}&\coloneqq (p\tri p^{(k)})\times\yon\\
\intertext{and the maps $f^{(k)}\colon p^{(k+1)}\to p^{(k)}$ are defined inductively as follows:}
p^{(0)}=p\times\yon&\To{f^{(0)}\coloneqq\tn{proj}}\yon=p^{(0)}&p^{(k+1)}=(p\tri p^{(k+1)})\times\yon&\To{f^{(k+1)}\coloneqq(p\tri f^{(k)})\times\yon}(p\tri p^{(k)})\times\yon=p^{(k+1)}
\end{align*}

\begin{proposition}\label{prop.free_cart_mon_cofree}
Let $\yon\to p$ be a cartesian map of polynomials, and let $\free{p}$ be the free monad on it; let $\cofree{p}$ be the cofree comonad on $p$. There is a natural cartesian monomorphism 
\[
\free{p}\to\cofree{p}\tri(\yon+1).
\]
\end{proposition}
\begin{proof}
    For any $q:\poly$, the functor $-\tri q$ commutes with limits. Since $\cofree{-}$ is a right adjoint, it also commutes with limits. Since the limit of cartesian monomorphisms is a cartesian monomorphism, we may assume $p$ is finitary and it suffices to produce compatible cartesian monomorphism $\varphi_i\colon p\tripow{i}\to p^{(i)}\tri(\yon+1)$ for each $i:\nn$. We take $\varphi_0$ to be the cartesian monomorphism $\yon\to\yon+1$. Suppose given $\varphi_i$. To define $p\tripow{i+1}\to(p\tri p^{(i)})\yon\tri(\yon+1)$, it suffices to give two maps, $p\tripow{i+1}\to p\tri p^{(i)}\tri(\yon+1)$ and $p\tripow{i+1}\to\yon+1$. For the latter, use $p_{(i+1)}\to 1\to\yon+1$. It remains to give a cartesian monomorphism $p\tripow{i+1}\to p\tri p^{(i)}\tri(\yon+1)$, which we obtain by induction $p\tri p\tripow{i}\To{p\tri\varphi_i}p\tri p^{(i)}\to p\tri p^{(i)}\tri(\yon+1)$.
\end{proof}

\begin{lemma}\label{lemma.cart_mon_coprod_incl}
    In $\poly$, any cartesian monomorphism $p\to q$ is a coproduct inclusion.
\end{lemma}
\begin{proof}
    Suppose $\varphi\colon p\to q$ is a cartesian monomorphism. Since $p\mapsto p(1)$ is a right adjoint, $\varphi(1)\colon p(1)\to q(1)$ is an injection; let $J'\coloneqq q(1)-p(1)$ be its compliment. Then we have the desired isomorphism
    \[
q\cong p+\sum_{j:J'}\yon^{q[j]}.
    \]
\end{proof}

\begin{proposition}\label{prop.cofree_free}
For any pointed polynomial $\yon\to p$ there is a map
$\cofree{p}\tri(\yon+1)\to\free{p}+1$.
\end{proposition}
\begin{proof}
    By Proposition \ref{prop.free_cart_mon_cofree} we have a cartesian monomorphism $\free{p}\to\cofree{p}\tri\yon+1$, which is a coproduct inclusion by Lemma \ref{lemma.cart_mon_coprod_incl}. For any coproduct $p\to q$ inclusion in a category with terminal object, we have a map $q\cong p+p'\to p+1$, completing the proof.
\end{proof}

Let's return to our goal of producing an $\iter$ map as in \eqref{eqn.iter}. By \eqref{eqn.cofree_adjunction}, it suffices to define a polynomial map of the form $\cofree{[A\yon,t\tri(A+B)\yon]}\to[A\yon,t\tri B\yon]$, or equivalently one of the form $A\yon\otimes\cofree{[A\yon,t\tri(A+B)\yon]}\to t\tri B\yon$. Given the exceptional structure $\xi\colon t+1\to t$ on $t$ and the fact that any monad $t$ carries an algebra structure $\free{t}\to t$, it suffices by Proposition \ref{prop.cofree_free} to find a map $A\yon\otimes\cofree{[A\yon,t\tri(A+B)\yon]}\to \cofree{t}\tri (\yon+1)\tri B\yon$.

We produce the desired map again by induction. The right-hand side is the limit
\[
\cofree{t}\tri (\yon+1)\tri B\yon\cong\lim(t^{(i)}\tri(\yon+1)\tri B\yon).
\]
When $i=0$, we use $A\yon\otimes\cofree{[A\yon,t\tri(A+B)\yon]}\to 1\to t^{(0)}\tri (\yon+1)\tri B\yon$.
 Suppose given a map $\varphi_i\colon A\yon\otimes\cofree{[A\yon,t\tri(A+B)\yon]}\to t^{(i)}\tri(\yon+1)\tri B\yon$. Define
\begin{align*}
    A\yon\otimes\cofree{[A\yon,t\tri(A+B)\yon]}&\to
\yon\times(A\yon\otimes([A\yon,t\tri(A+B)\yon]\tri\cofree{[A\yon,t\tri(A+B)\yon]})\\&\to
\yon\times(t\tri(A+B)\yon\tri\cofree{[A\yon,t\tri(A+B)\yon]})\\&\to
\yon\times(t\tri ((A\yon\tri\cofree{[A\yon,t\tri(A+B)\yon]})+B\yon))\\&\To{\cong}
\yon\times(t\tri ((A\yon\otimes\cofree{[A\yon,t\tri(A+B)\yon]})+B\yon))\\&\to
\yon\times(t\tri((t^{(i)}\tri(\yon+1)\tri B\yon)+B\yon))\\&\to
\yon\times(t\tri t^{(i)}\tri(\yon+1)\tri B\yon)\\&\to
(\yon\times(t\tri t^{(i)}))\tri(\yon+1)\tri B\yon)\\&\To{\cong}
t^{(i+1)}\tri(\yon+1)\tri B\yon
\end{align*}

We have now defined the purported $\iter$ map and hence trace map, as explained above.

\begin{conjecture}
The purported trace on $\dk$ defined above satisfies the axioms of a traced monoidal category.
\end{conjecture}

While this conjecture has not been proven, it has influenced the development of a working implementation in the open source  \href{https://github.com/AlgebraicJulia/AlgebraicRewriting.jl}{AlgebraicRewriting.jl}. This is the subject of \cref{chap.application}.

\section{Application to rewriting and agent-based modeling}\label{chap.application}

\subsection{Attributed C-Sets}

Our case study uses the AlgebraicJulia ecosystem \cite{algjulia} due to its support for wiring diagram manipulation (Catlab.jl) as well as its graph rewriting library, AlgebraicRewriting.jl \cite{rewrite}. The core data structure of AlgebraicJulia is the ACSet, i.e.\ attributed $\cat{C}$-Set for some finitely-presented category $\cat{C}$. ACSets offer a category-theoretic model of databases which extends $\cat{C}$-Sets (i.e.\ copresheaves) to include noncombinatorial data \cite{Patterson2022categoricaldata}. The database schema is given by a profunctor, i.e.\ a functor $S:|S|\to 2$, which distinguishes objects in $|S|$ as representing either tables or attribute types. Given an assignment $K: S_1 \to \smset$ which provides concrete Julia types for attributes, the category $\acset^S_K$ is bicomplete and a topos, and thus is an appropriate setting for applying graph rewriting rules. 

Consider an ACSet $X$, equipped with a distinguished `focus', i.e.\ morphism $A \to X$. We will soon see applications where it makes sense to think of $A$ as the \emph{shape} of a particular agent in the state of the world $X$, where the agent is picked out by the chosen morphism. Note that considering ACSets without any agent is tantamount to picking the agent shape to be an empty ACSet, 0, which is the initial object of $\acset^S_K$.

For example, consider the task of modeling wolves, sheep, and grass distributed on a directed graph, where the wolves and sheep are facing particular directions and have some integer number of energy units. Grass grows on vertices and has an integer number of days until it is grown. The schema in Figure~\ref{fig:lv} shows one way to model this.

\begin{figure}[!h]

\begin{minipage}{.4\textwidth}

\[\begin{tikzcd}
	& \textit{Dir} \\
	& E \\
	\textit{Wolf} & V & \textit{Sheep} \\
	\\
	& {\nn}
	\arrow["\mathit{tgt}", curve={height=-6pt}, from=2-2, to=3-2]
	\arrow["\mathit{src}"', curve={height=6pt}, from=2-2, to=3-2]
	\arrow["{w_{\mathit{eng}}}"', dashed, from=3-1, to=5-2]
	\arrow["{w_{\mathit{pos}}}"', from=3-1, to=3-2]
	\arrow["{s_{\mathit{pos}}}", from=3-3, to=3-2]
	\arrow["{s_{\mathit{eng}}}", dashed, from=3-3, to=5-2]
	\arrow["{\footnotesize \textit{grass}}"{description}, dashed, from=3-2, to=5-2]
	\arrow["{w_{\mathit{dir}}}", dashed, from=3-1, to=1-2]
	\arrow["{s_{\mathit{dir}}}"',dashed, from=3-3, to=1-2]
	\arrow["{\footnotesize dir}"{description}, dashed, from=2-2, to=1-2]
\end{tikzcd}\]
\end{minipage}%
\begin{minipage}{.6\textwidth}
    \includegraphics[width=\textwidth, right]{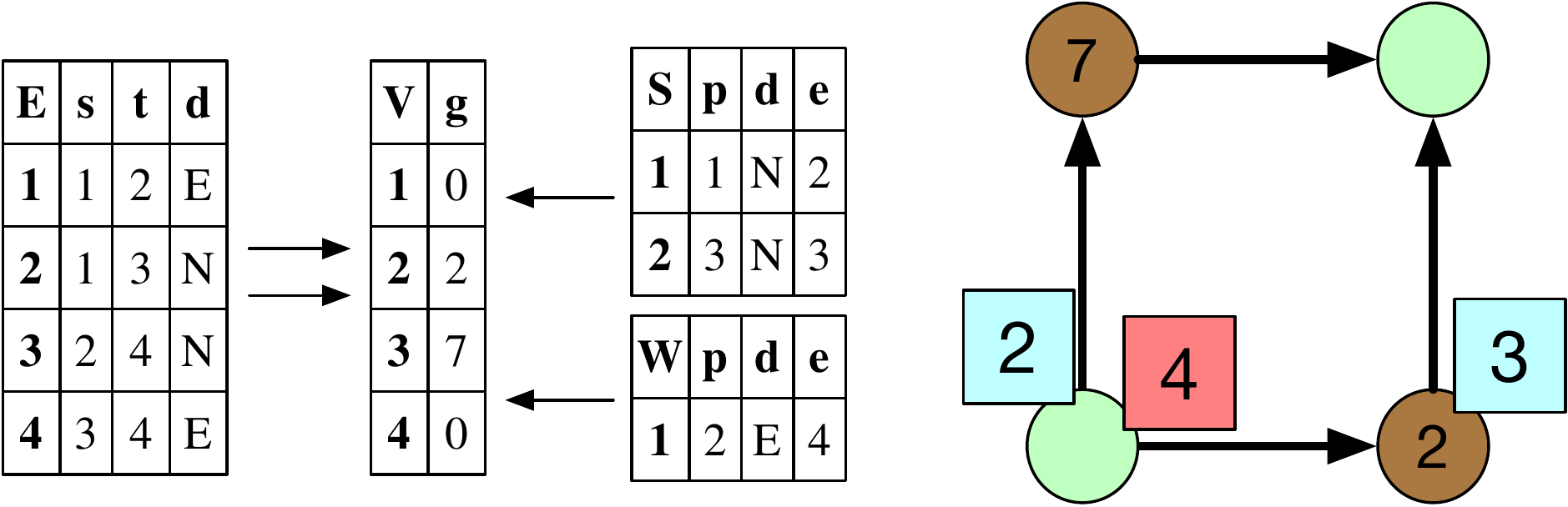}
\end{minipage}

    \caption{\textbf{Left} An ACSet schema with six objects and morphisms as well as two attribute types and six attributes (dotted edges). \textbf{Center} An example ACSet on this schema \textbf{Right} The example, informally visualized with energy as integers, sheep as blue boxes, and wolves as red boxes. If we wished to distinguish the wolf as the `agent', we would use an ACSet morphism into this example from an ACSet with one wolf, zero edges, and one vertex, i.e.\ the \emph{shape} of a wolf.}
    \label{fig:lv}
\end{figure}

\subsection{A DSL for graph rewriting programs}

We use the theory of $\dk$ to implement a graph rewriting programming language, with data manipulation specified by rewrite rules acting upon ACSets (possibly with agents). These programs can be assembled from a small number of primitive generators. In conjunction with Catlab's general infrastructure for manipulating wiring diagrams, these primitives function as a powerful domain-specific language for agent-based modeling and programming via graph rewriting.

The domain and codomain of our morphisms of interest consist of sets of diagrams of the form in Figure~\ref{fig:state}, and coproducts thereof. Each such diagram, which we will call a \textit{trajectory}, is a sequence of `world states' $X_i$ with distinguished focuses $A_i\to X_i$. If the trajectory is represented by a variable \lstinline{t::Traj} in pseudocode, then let \lstinline{last(t)} return $X_n$, \lstinline{length(t)} return $n$, and \lstinline{t[i]} return $A_i \to X_i$. Let \lstinline{postcompose(t::Traj,f::Hom(X,Xi),i::Int)} compute the composite of \lstinline{f} with the partial maps from $i$ to $n$. This returns either a total morphism or \lstinline{nothing}. Let \lstinline{(t::Traj)+(b::Hom(A,Xn))} extend the trajectory with an identity partial map. Because there are now infinite sets as the domains and codomains of our morphisms, we adopt the following shorthand: when visualizing wiring diagrams, a wire labeled by an ACSet $A_n$ represents the set of all trajectories whose current agent is $A_n$. Stating a morphism is of type $A + B \rightarrow C$ for ACSets $A,B,C$ indicates the domain is the coproduct of the set of trajectories with current agent $A$ and the trajectories with current agent $B$ and that the codomain is the set of trajectories with current agent $C$.

\begin{figure}[!h]
    \centering

\[\begin{tikzcd}
	{A_1} & {A_2} & {\cdots} & {A_n} \\
	{X_1} & {X_2} & {\cdots} & {X_n}
	\arrow["\shortmid"{marking}, from=2-1, to=2-2]
	\arrow[from=1-1, to=2-1]
	\arrow[from=1-2, to=2-2]
	\arrow[from=1-4, to=2-4]
	\arrow["\shortmid"{marking}, from=2-3, to=2-4]
	\arrow["\shortmid"{marking}, from=2-2, to=2-3]
\end{tikzcd}\]
    \caption{A trajectory in the space of ACSets. $X_1$ and $X_n$ respectively represent the initial (resp. current) state of the world during the simulation, and each successive world state is related to the previous via a partial map (indicated by a ticked arrow). Each world state $X_i$ also has a distinguished focus $A_i\to X_i$.}
    \label{fig:state}
\end{figure}

Some generating morphisms for rewriting programs are shown in Figures~\ref{fig:icon1}~and~\ref{fig:icon2}. The semantics of the stateless generators is visually represented in Figure \ref{fig:icon1} and described here: \lstinline{Rewrite} extends a trajectory with a partial map induced by applying the rewrite rule (DPO, SPO, SqPO, and PBPO+ \cite{overbeek2021graph} semantics supported). Rewrite rules must also have their pattern $L$ and replacement $R$ related to a specific input agent shape $A$ and output agent shape $B$, respectively. The input agent shape imposes a strong constraint on valid matches via a triangle which must commute. If, nevertheless, multiple matches are valid, an arbitrary one is selected. If it is successfully rewritten, the $B$ outport is exited, otherwise the $A$ outport is exited. \lstinline{Weaken} extends a trajectory without changing the state of the world $X_n$ by precomposing the agent morphism. \lstinline{Strengthen} extends a trajectory via pushout, which simultaneously changes the agent shape and the state of the world. \lstinline{Init} switches the trajectory to a particular world state and agent, with no relation to the previous world state. \lstinline{Fail} can be given the semantics of raising an exception. Alternatively, if the context is the $\List+1$ monad, it could  silently produce an empty list.

The semantics of the \lstinline{ControlFlow} box is to redirect an input to one of its outports, possibly nondeterministically and possibly as a function of its trajectory data. These morphisms are of the form $\Hom_{\dk}(A,n\times A)$ for some set $n$ and have a Mealy transition function which is \textit{pure}, by which we mean it can be factored into a map $\phi: S \times A \rightarrow S \times (n\times A)$ (such that $\phi\then\pi_3=\pi_2$) followed by the monad unit $\eta$. 

\[\begin{tikzcd}
	{S\times A} & {S\times (n\times A)} \\
	& {S\times t\triangleleft(n\times A)}
	\arrow["{\text{pure}\ \phi}"', from=1-1, to=2-2]
	\arrow["\phi", from=1-1, to=1-2]
	\arrow["{S \times \eta\triangleleft(n\times A)}", from=1-2, to=2-2]
\end{tikzcd}\]

The semantics of the \lstinline{Query} box is a Mealy machine with state space $\List(\Hom_\acset) \times \nn$. The first element is the list of queued `agents' we have yet to process, and the second element keeps track of what time step (in the trajectory) the \lstinline{Query} box was originally entered. The dynamics are given by two functions, an update function $S \times A + C \to S$ and a readout function $S \times A + C \to A + B + 0$. Entering through the $A$ port at step $n$ has the significance of starting the \lstinline{Query} process anew; the internal state is overwritten to store $n$ and a list of morphisms $B \to X_n$. Entering the $C$ port communicates that we have finished one of the subagent's subroutines - the agent is removed from the box's state and all other agents are pushed forward to the current timestep. How the \lstinline{Query} box is exited is firstly determined by whether or not there are any remaining agents to process; if there are any, the $B$ port is exited with a $B$ agent. If there are none, we try to exit with original $A$ agent. If this agent is total when brought to the current state, we exit through the $A$ port, otherwise the $0$ port. In pseudocode, these functions are characterized below:

\begin{lstlisting}
updateA(_, traj::Traj) = (homomorphisms(B,last(traj)), length(traj))
updateC((bnext:bs,i), traj::Traj) = 
  ([b for b in bs if !isnothing(postcompose(traj,b,i))],i)
updateC(([],i::Int),_) = error

readout((bnext:bs,_), traj::Traj) = return (B,traj + bnext)
readout(([],i::Int), traj::Traj) = case postcompose(traj,traj[i],i) of 
  nothing => return (0, traj + initial(last(traj)))
  new_a_x => return (A, traj + new_a_x)
\end{lstlisting}

\begin{figure}[!h]
    \centering
    \includegraphics[width=0.9\textwidth]{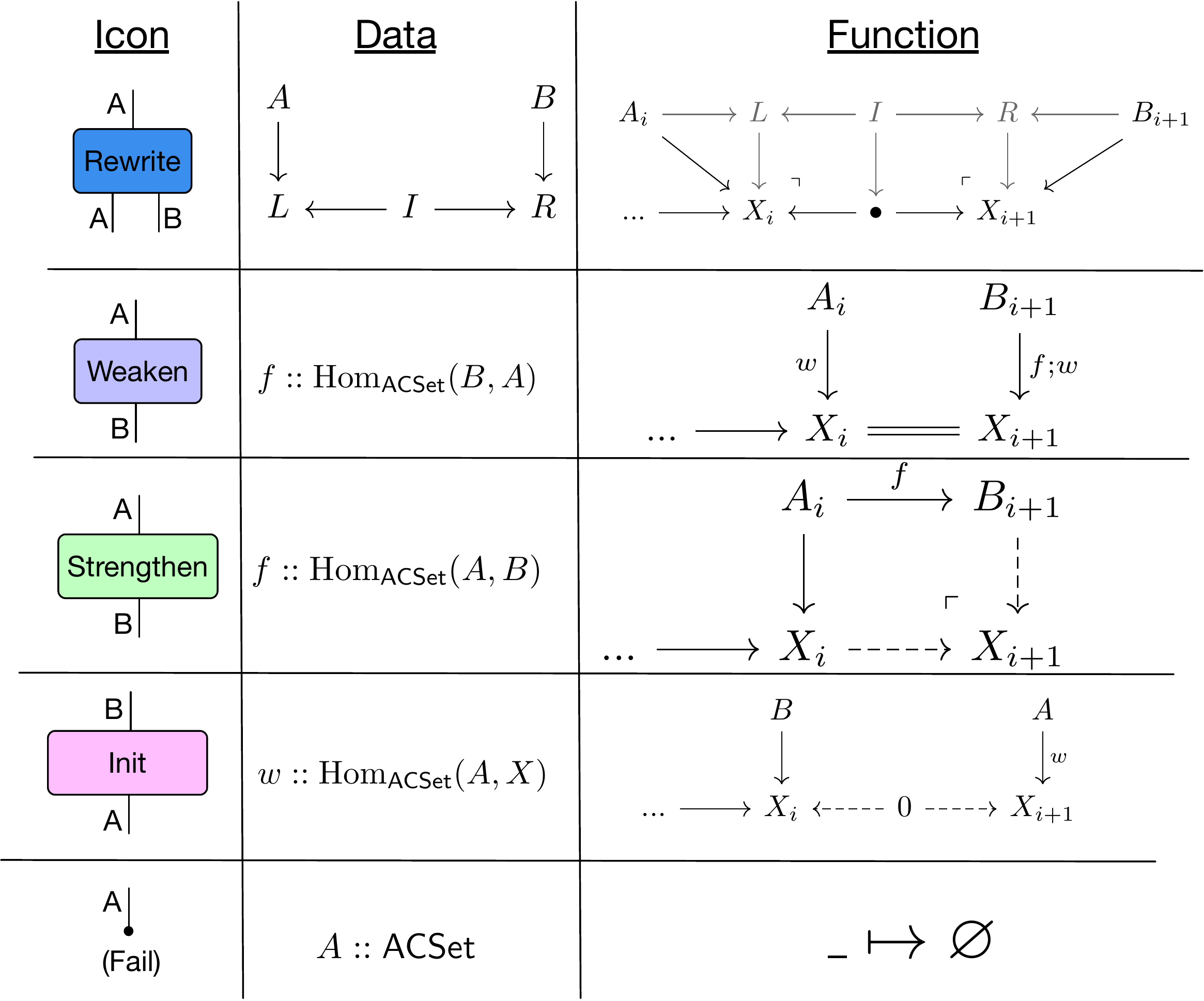}
    \caption{Primitive generating morphisms for graph rewriting programs which have trivial dynamics, i.e.\ they are pure functions from input ports to output ports, with exception to \lstinline{Fail}, which can raise an exception.}
    \label{fig:icon1}
\end{figure}

\begin{figure}[!h]
    \centering

    \includegraphics[width=0.9\textwidth]{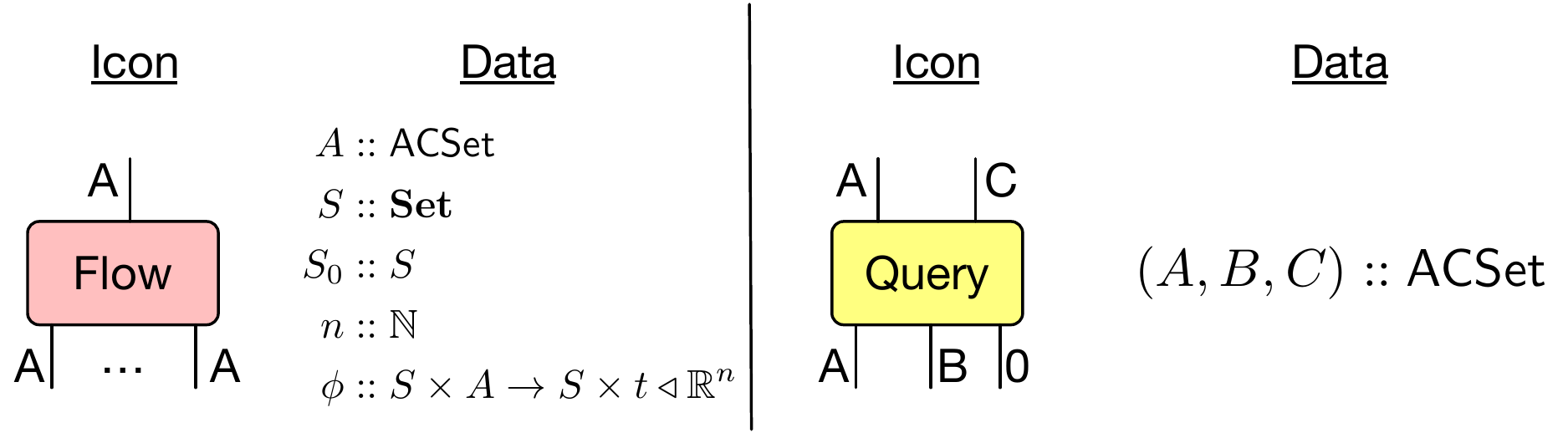}
    \caption{Primitive generating morphisms for graph rewriting programs which have nontrivial dynamics. \textbf{Left.} A generator for pure control flow. The value $\mathbb{R}^{\geq0}$ assigned to each outgoing wire is a weight to bias the probability of exiting through that port. \textbf{Right.} A generator for running a subroutine $B \to C$ for each $B$ agent in the the current world, $X_i$.}
    \label{fig:icon2}
\end{figure}

Although the `standard' means of using a \lstinline{Query} box is to connect a subroutine from the $B$ outport to the $C$ inport, yielding a $A \to A + 0$ interface, it can be used in more flexible ways. For example, a procedure which applies the rule $rw$ to a single (arbitrary) $A$ in the world state which satisfies property $\phi$ is visualized in Figure \ref{fig:q}.

\begin{figure}[!h]
    \centering
    \includegraphics[width=0.6\textwidth]{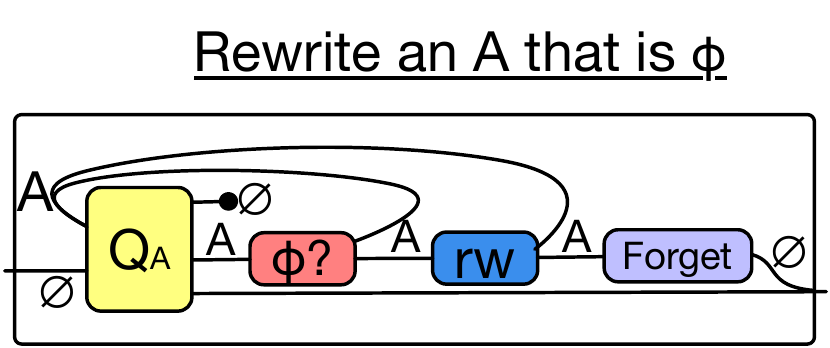}
    \caption{A program $0\to 0$, in context of the $\Maybe$ monad. Its input is a trajectory with an empty agent in its current time step. We then loop over possible agents of shape $A$. As soon as one is found which both satisfies property $\phi$ \textit{and} is successfully rewritten, we replace our focus $A$-shaped agent with the unique $0$-shaped agent and exit.}
    \label{fig:q}
\end{figure}

The above primitives were designed in order to be both easily interpretable (e.g.\ control flow, focus shifting, and state changing are all separated concerns) as well as expressive enough to reproduce popular agent-based models. However, these primitives can be extended in a principled rather than ad-hoc manner, due to the specification that new primitives must contain the data of a morphism in $\dk$. Furthermore, these primitives can be composed to libraries of operations at a higher level of abstraction for use by a domain expert.

\subsection{Implementation}

Our implementation is general over a finitary exceptional polynomial monad as in Definition \ref{def.exceptional}, using a codata structure for representing (potentially-infinite) behavior trees. These trees can be provided directly for each generator, although it is more convenient to programmatically generate them from a Mealy machine representation. By switching from $t=\Maybe$ to $\List+1$ or $\Dist+1$, we obtain nondeterministic or probabilistic simulations. For many graph rewriting examples, we wish to consider all possible matches, not merely an arbitrary match. A principled way to do this is to use a \lstinline{Rewrite} which produces a list of outputs, corresponding to all possible matches. This is particularly important if a program is being constructed to empirically test which graphs are reachable via a collection of rewrite rules. Furthermore, if certain matches are more likely than others, then a distribution on this output list can be incorporated into our programs.

One way in which we can take advantage of our formalism for rewriting programs is functorial data migration \cite{spivak2012functorial}. Given a functor $S\to T$ between ACSet schemas, a $\Sigma$ migration pushes $S$ ACSets forward to $T$ ACSets in a universal way, while $\Delta$ migration migrates data the other direction. As this is functorial, it is possible to apply these migrations to ACSet morphisms, rewrite rules, and entire graph rewriting programs. Another key advantage of working graphically is the various forms of composition of wiring diagrams, which allow for concise operations for organization of morphisms with $\otimes$ and $\then$ as well as hierarchically constructing programs via operadic substitution. 

\subsection{Example: discrete Lotka Volterra model}

The first example agent-based model showcased by Netlogo \cite{tisue2004netlogo} on their website is a model of wolf-sheep predation. An analogue of this model using the present framework is presented in Figure~\ref{fig:model}. Its construction leverages $\Sigma$ and $\Delta$ data migrations, operadic substitution, $t=\Maybe$, and the primitives \lstinline{Rewrite}, \lstinline{Weaken}, \lstinline{ControlFlow}, and \lstinline{Query}. A pedagogical walkthrough of the model's construction and running the model is provided in an accompanying \href{https://nbviewer.org/github/AlgebraicJulia/AlgebraicRewriting.jl/blob/compat_varacsets/docs/src/Dynamic%20Tracing.ipynb}{notebook}. An example construction in $\dk$ using the $\Dist$ monad can also be found in the accompanying notebook.

\begin{figure}[!h]
    \centering
    \includegraphics[width=\textwidth]{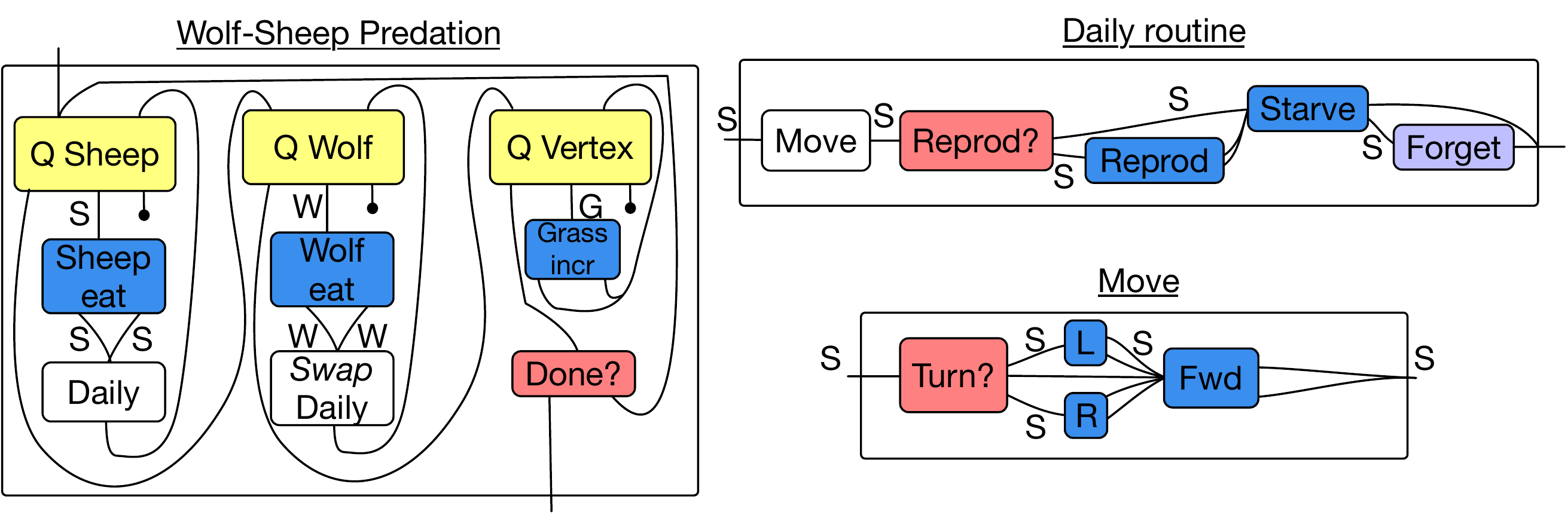}
    \caption{A program performing a wolf-sheep simulation. $\mathit{Swap}$ here represents a $\Delta$ migration which swaps wolves and sheep in the schema of Figure \ref{fig:lv}, allowing for a shared implementation of actions that are common to wolves and sheep. The agent shape $0$ is depicted as an unlabeled wire, although wires sharing the same target port need not all be labeled, as they must share the same agent for the program to be well-typed.}
    \label{fig:model}
\end{figure}

\section{Conclusion and Future work}\label{chap.conclusion}

A general theory of dynamical tracing guided the development of a graph rewriting DSL. Because this language is understood mathematically and expressed as combinatorial data, rather than programming syntax, high level rewrite procedures can be understood and serialized, independently from any particular implementation (although an implementation in Julia was developed). Abstract operations like data migration and various forms of composition become natural to perform on these procedures due to their interpretation as morphisms in $\dk$. Furthermore, the assimilation of Mealy machines, system interfaces and behaviors, as well as monadic effects via polynomial functors inspired an implementation that is both concise and general.

There are dimensions along which this work can be extended. The formalism we presented focuses on a dynamical system as something that is interacted with by a single agent. It remains frozen as the agent interacts with other systems, but this assumption could be relaxed as we consider multiple simultaneous agents---i.e.\ parallel programming with graph transformation. The notion of a trajectory presented here requires all ACSets involved to share the same schema (as there are no morphisms between ACSets of different schemas); this could be generalized to allow for multiscale modeling (data transformation in both a high level, `macroscopic' schema as well as a low level `microscopic' schema).

We plan to represent more existing agent-based models as well as develop new ones in this formalism. Agent-based models are a preferred style of modeling in any situation with emergent effects, such as physical phenomena (e.g.\ flow and diffusion simulations) \cite{bonabeau2002agent}, human transportation networks \cite{smith1995transims},
and epidemiology \cite{hunter2017taxonomy}. Adding Catlab support for incremental graph matching \cite{fan2013incremental} will be important for rewriting programs to be competitive in performance with established software like Netlogo and Kappa.

\bibliographystyle{splncs04}
\bibliography{references}

\begin{thebibliography}{10}
\providecommand{\url}[1]{\texttt{#1}}
\providecommand{\urlprefix}{URL }
\providecommand{\doi}[1]{https://doi.org/#1}

\bibitem{Ahman_2016}
Ahman, D., Uustalu, T.: Directed containers as categories. Electronic
  Proceedings in Theoretical Computer Science  \textbf{207},  89--98 (apr
  2016). \doi{10.4204/eptcs.207.5}, \url{https://doi.org/10.4204%2Feptcs.207.5}

\bibitem{blostein1996issues}
Blostein, D., Fahmy, H., Grbavec, A.: Issues in the practical use of graph
  rewriting. In: Graph Grammars and Their Application to Computer Science: 5th
  International Workshop Williamsburg, VA, USA, November 13--18, 1994 Selected
  Papers 5. pp. 38--55. Springer (1996)

\bibitem{bonabeau2002agent}
Bonabeau, E.: Agent-based modeling: Methods and techniques for simulating human
  systems. Proceedings of the national academy of sciences
  \textbf{99}(suppl\_3),  7280--7287 (2002)

\bibitem{rewrite}
Brown, K., Patterson, E., Fairbanks, J.P.: Double pushout rewriting of c-sets.
  CoRR  \textbf{abs/2111.03784} (2021), \url{https://arxiv.org/abs/2111.03784}

\bibitem{bunke1982attributed}
Bunke, H.: Attributed programmed graph grammars and their application to
  schematic diagram interpretation. IEEE Transactions on Pattern Analysis and
  Machine Intelligence (6),  574--582 (1982)

\bibitem{corradini2006sesqui}
Corradini, A., Heindel, T., Hermann, F., K{\"o}nig, B.: Sesqui-pushout
  rewriting. In: International Conference on Graph Transformation. pp. 30--45.
  Springer (2006)

\bibitem{ehrig1973graph}
Ehrig, H., Pfender, M., Schneider, H.J.: Graph-grammars: An algebraic approach.
  In: 14th Annual Symposium on Switching and Automata Theory (swat 1973). pp.
  167--180. IEEE (1973)

\bibitem{fahmy1996reasoning}
Fahmy, H.: Reasoning in the presence of uncertainty via graph rewriting.
  (1996)

\bibitem{fan2013incremental}
Fan, W., Wang, X., Wu, Y.: Incremental graph pattern matching. ACM Transactions
  on Database Systems (TODS)  \textbf{38}(3),  1--47 (2013)

\bibitem{hunter2017taxonomy}
Hunter, E., Mac~Namee, B., Kelleher, J.D.: A taxonomy for agent-based models in
  human infectious disease epidemiology. Journal of Artificial Societies and
  Social Simulation  \textbf{20}(3) (2017)

\bibitem{joyal1996traced}
Joyal, A., Street, R., Verity, D.: Traced monoidal categories. In: Mathematical
  proceedings of the cambridge philosophical society. vol.~119, pp. 447--468.
  Cambridge University Press (1996)

\bibitem{kelly1982basic}
Kelly, M.: Basic concepts of enriched category theory, vol.~64. CUP Archive
  (1982)

\bibitem{lowe1993algebraic}
L{\"o}we, M.: Algebraic approach to single-pushout graph transformation.
  Theoretical Computer Science  \textbf{109}(1-2),  181--224 (1993)

\bibitem{niupolynomial}
Niu, N., Spivak, D.I.: Polynomial functors: A general theory of interaction

\bibitem{overbeek2021graph}
Overbeek, R., Endrullis, J., Rosset, A.: Graph rewriting and relabeling with
  pbpo+. In: Graph Transformation: 14th International Conference, ICGT 2021,
  Held as Part of STAF 2021, Virtual Event, June 24--25, 2021, Proceedings. pp.
  60--80. Springer (2021)

\bibitem{algjulia}
Patterson, E., other contributors: Algebraicjulia/catlab.jl: v0.13.5 (Dec
  2021). \doi{10.5281/zenodo.5771194},
  \url{https://doi.org/10.5281/zenodo.5771194}

\bibitem{Patterson2022categoricaldata}
Patterson, E., Lynch, O., Fairbanks, J.: Categorical {D}ata {S}tructures for
  {T}echnical {C}omputing. {Compositionality}  \textbf{4} (Dec 2022).
  \doi{10.32408/compositionality-4-5},
  \url{https://doi.org/10.32408/compositionality-4-5}

\bibitem{Patterson_2021}
Patterson, E., Spivak, D.I., Vagner, D.: Wiring diagrams as normal forms for
  computing in symmetric monoidal categories. Electronic Proceedings in
  Theoretical Computer Science  \textbf{333},  49--64 (feb 2021).
  \doi{10.4204/eptcs.333.4}, \url{https://doi.org/10.4204%2Feptcs.333.4}

\bibitem{plump2012design}
Plump, D.: The design of gp 2. arXiv preprint arXiv:1204.5541  (2012)

\bibitem{schurr1991progress}
Sch{\"u}rr, A.: Progress: A vhl-language based on graph grammars. In: Graph
  Grammars and Their Application to Computer Science: 4th International
  Workshop Bremen, Germany, March 5--9, 1990 Proceedings 4. pp. 641--659.
  Springer (1991)

\bibitem{Selinger_2010}
Selinger, P.: A survey of graphical languages for monoidal categories. In: New
  Structures for Physics, pp. 289--355. Springer Berlin Heidelberg (2010).
  \doi{10.1007/978-3-642-12821-9_4},
  \url{https://doi.org/10.1007%2F978-3-642-12821-9_4}

\bibitem{smith1995transims}
Smith, L., Beckman, R., Baggerly, K.: Transims: Transportation analysis and
  simulation system. Tech. rep., Los Alamos National Lab.(LANL), Los Alamos, NM
  (United States) (1995)

\bibitem{4354063}
(https://math.stackexchange.com/users/666875/richard southwell), R.S.: What is
  a coproduct in an enriched category? Mathematics Stack Exchange,
  \url{https://math.stackexchange.com/q/4353966},
  uRL:https://math.stackexchange.com/q/4353966 (version: 2022-01-11)

\bibitem{spivak2012functorial}
Spivak, D.I.: Functorial data migration. Information and Computation
  \textbf{217},  31--51 (2012)

\bibitem{poly}
Spivak, D.I.: A reference for categorical structures on $\mathbf{Poly}$ (2022).
  \doi{10.48550/ARXIV.2202.00534}, \url{https://arxiv.org/abs/2202.00534}

\bibitem{tisue2004netlogo}
Tisue, S., Wilensky, U.: Netlogo: A simple environment for modeling complexity.
  In: International conference on complex systems. vol.~21, pp. 16--21.
  Citeseer (2004)

\bibitem{voss2023graph}
Voss, C., Petzold, F., Rudolph, S.: Graph transformation in engineering design:
  an overview of the last decade. AI EDAM  \textbf{37}, ~e5 (2023)

\bibitem{wisnesky2011minimizing}
Wisnesky, R.: Minimizing monad comprehensions. Tech. rep., Citeseer (2011)

\end{thebibliography}
\end{document}